\documentclass[11pt]{article}
\usepackage{amssymb,amsthm,amsmath}
\usepackage{mathabx}
\usepackage{graphicx}
\usepackage{epsfig,color}
\usepackage{epsf}
\usepackage[small]{caption}
\usepackage{amsfonts}
\usepackage{fullpage}
\usepackage{enumerate}
\usepackage{enumitem}
\usepackage{url}

\newtheorem{theorem}{Theorem}
\newtheorem{lemma}{Lemma}

\newtheorem{corollary}{Corollary}

\newcommand{\RR}{\mathbb{R}} %  set of real numbers
\newcommand{\eps}{\varepsilon}

\newcommand{\conv}{{\rm conv}}
\newcommand{\diam}{{\rm diam}}
\newcommand{\len}{{\rm len}}
\newcommand{\area}{{\rm Area}}
\newcommand{\per}{{\rm per}}
\newcommand{\vol}{{\rm Vol}}

\newcommand{\opt}{\textsf{OPT}}

\newcommand{\etal}{{et~al.}}
\newcommand{\ie}{{i.e.}}
\newcommand{\eg}{{e.g.}}

\newcommand{\later}[1]{}
\newcommand{\old}[1]{}

\providecommand{\intd}[0]%
{\;\mbox{d}}

\title{{\sc On the Shortest Separating Cycle}\thanks{A preliminary
    version of this paper appeared in the {\em Proceedings of the 29th
      Canadian Conference on Computational Geometry}, Ottawa, ON,
    Canada, July 2017.
}}

\author{
Adrian Dumitrescu\thanks{%
Department of Computer Science,
University of Wisconsin--Milwaukee, USA\@.
Email:~\texttt{dumitres@uwm.edu}}}
\begin{document}

\maketitle

\begin{abstract}
 According to a result of Arkin~\etal~(2016), given $n$ point pairs in the plane,
  there exists a simple polygonal cycle that separates the two points in each pair
  to different sides; moreover, a $O(\sqrt{n})$-factor approximation 
  with respect to the minimum length can be computed in polynomial time.
  Here the following results are obtained:
  
  (I)~We extend the problem to geometric hypergraphs and obtain the following characterization of 
  feasibility. Given a geometric hypergraph on points in the plane with hyperedges of size at least $2$, 
  there exists a simple polygonal cycle that separates each hyperedge if and only if the
  hypergraph is $2$-colorable.

  (II)~We extend the $O(\sqrt{n})$-factor approximation in the length measure as follows:
  Given a geometric graph $G=(V,E)$, a separating cycle (if it exists) can be computed
  in $O(m+ n\log{n})$ time, where $|V|=n$, $|E|=m$. 
  Moreover, a $O(\sqrt{n})$-approximation of the shortest separating cycle can be found
  in polynomial time.
  Given a geometric graph $G=(V,E)$ in $\RR^3$, a separating polyhedron (if it exists)
  can be found in $O(m+ n\log{n})$ time, where $|V|=n$, $|E|=m$. 
  Moreover, a $O(n^{2/3})$-approximation of a separating polyhedron of minimum perimeter
  can be found in polynomial time.
  
  (III)~Given a set of $n$ point pairs in convex position in the plane, we show that
  a $(1+\eps)$-approximation of a shortest separating cycle can be computed in time
  $n^{O(\eps^{-1/2})}$. In this regard, we prove a lemma on convex polygon approximation
  that is of independent interest. 


\end{abstract}

\section{Introduction} \label{sec:intro}

Given a set of $n$ pairs of points in the plane with no common elements,
$\{(p_i,q_i) \ | \ i=1,\ldots,n\}$, a \emph{shortest separating cycle} is a plane cycle
(a closed curve, a.k.a. \emph{tour}) of minimum length that contains inside exactly one point
from each of the $n$ pairs. The problem {\sc Shortest Separating Cycle} is that
of finding such a cycle, given the input pairs. It was introduced by Arkin~\etal~\cite{AGH+16}
motivated by applications in data storage and retrieval in distributed sensor networks.
The authors gave a $O(\sqrt{n})$-factor approximation for the general case and better
approximations for some special cases. On the other hand, using a reduction from
{\sc Vertex Cover}, they showed that the problem is hard to approximate
for a factor of $1.36$ unless ${\rm P}={\rm NP}$, 
and is hard to approximate for a factor of $2$ assuming the Unique Games Conjecture;
see, \eg, \cite[Ch.~16]{WS11} for technical background.

The assumption that no point appears more than once,
\ie, $|\{p_1,\ldots,p_n\} \cup \{q_1,\ldots,q_n\}|=2n$, is sometimes necessary
for the existence of a separating cycle; \ie, there are instances of sets of pairs with
common elements and no separating cycle; see for instance Fig.~\ref{fig:f1} (the edges
in these graphs represent pairs of input points).
For convenience, points on the boundary of the cycle are considered inside;
it is easy to see that requiring points to lie strictly in the interior or
also on the boundary are equivalent variants in regards to the existence
of a separating cycle. Moreover, the equivalence is almost preserved in the length measure:
given any positive $\eps>0$, and a separating cycle $C$ for $n$ pairs,
enclosing $P=\{p_1,\ldots,p_n\}$ (after relabeling each pair, if needed),
with some of the points of $P$ on its boundary, a separating cycle of length at most
$(1+\eps) \, \len(C)$ can be constructed, having all points of $P$ in its interior.

In this paper we study the extension of the concept of separating cycle to
arbitrary graphs and hypergraphs, and to higher dimensions;
in the original version introduced by Arkin~\etal~\cite{AGH+16}, the input graph is a \emph{matching},
\ie,  it consists of $n$ edges with no common endpoints; see Fig.~\ref{fig:f3} for an example. 
Two instances with $8$ and respectively $3$ point pairs that do not admit separating cycles
are illustrated in Fig.~\ref{fig:f1}; the common reason is that both graphs contain odd cycles
and odd cycles do not admit separating cycles.
\begin{figure}[hbtp]
\centering\includegraphics[scale=0.9]{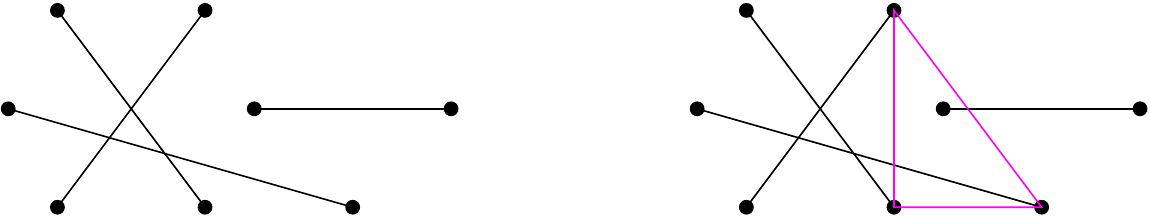}
\caption{A shortest separating cycle for a matching.}
\label{fig:f3}
\end{figure}

We observe that for arbitrary input graphs, one cannot use the algorithm from~\cite{AGH+16}.
That algorithm (in~\cite[Subsec.~3.5]{AGH+16}) first computes a minimum-size square $Q$
containing at least one point from each pair, and then computes a constant-factor
approximation of a shortest cycle (tour) of the points contained in $Q$, in the form of
a simple polygon. In the end, this tour is refined to a separating cycle of the given set
of point pairs with only a small increase in length.
Here we note that there exist instances, such as that in Fig.~\ref{fig:f1}\,(right),
for which there is no separating cycle confined to $Q$;
moreover, the length of a shortest separating cycle can be arbitrarily
larger than any function of $\diam(Q)$ and $n$, and so a new approach is needed
for the general version with arbitrary input graphs, or its extension to hypergraphs;
\ie, the current $O(\sqrt{n})$-factor approximation does not carry through to these settings.
\begin{figure}[hbtp]
\centering\includegraphics[scale=0.79]{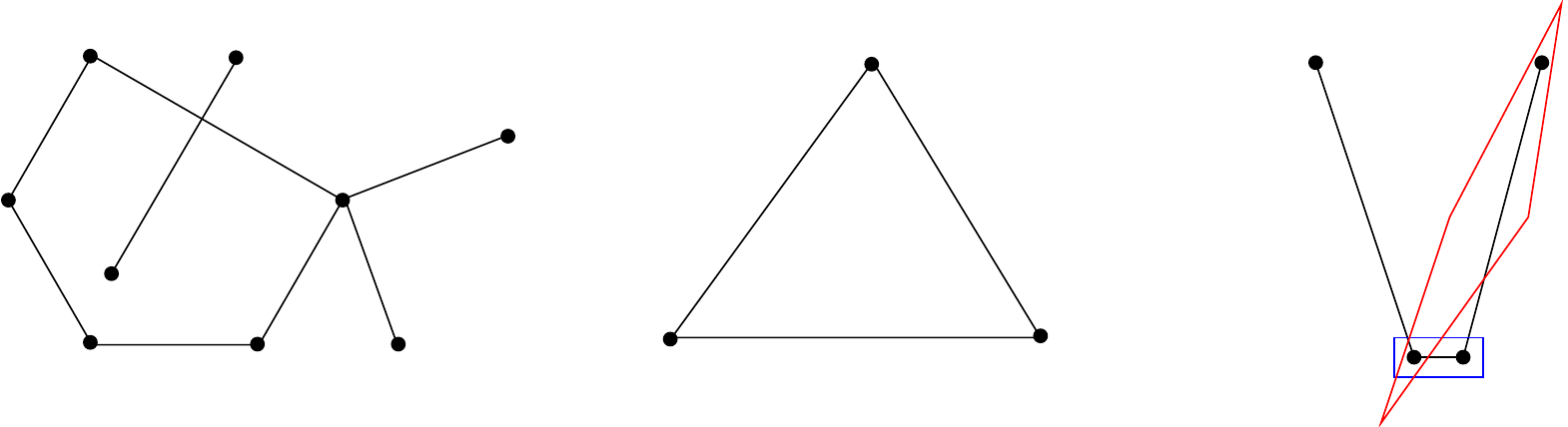}
%\centering\includegraphics[scale=0.79]{f1.eps}
\caption{Left and center: instances with no separating cycle.
  Right: instance where the minimum axis-parallel square (or rectangle)
  that contains at least one point from each pair does not lead to a solution;
a solution is indicated by the red cycle.}
\label{fig:f1}
\end{figure}

We first show that a planar geometric graph $G=(V,E)$ admits a separating cycle
(for all its edge-pairs) if and only if it is bipartite.
This result can be extended to hypergraphs in $\RR^d$.
Given a geometric hypergraph on points in $\RR^d$ with no singleton
edges\footnote{A \emph{singleton} edge is an edge with one vertex.},
there exists a simple polyhedron that separates each hyperedge if and only if the
hypergraph is $2$-colorable.

\paragraph{Definitions and notations.}
A \emph{hypergraph} is a pair $H=(V,E)$, where $V$ is finite set of \emph{vertices},
and $E$ is a family of subsets of $V$, called \emph{edges}.
$H$ is said to be $2$-\emph{colorable} if there is a $2$-coloring of $V$
such that no edge is monochromatic; see, \eg, \cite[Ch.~1.3]{AS15}. 

For a polygonal cycle $C$, let $\ring{C}$ and $\overline{C}$ denote the
interior and exterior of $C$, respectively; and $\partial C$ denote its boundary. 
Consider a geometric hypergraph $H=(V,E)$ on points in the plane with no singleton edges. 
A polygonal cycle $C$ is said to be a \emph{separating cycle} for $H$ if
(i)~$C$ is simple (\ie, with no self-intersections); and
(ii)~each edge of $H$ has points inside $C$
(in its interior or on its boundary) and points in the exterior of $C$; that is,
for each edge $A \in E$, both $A \cap (\ring{C} \cup \partial C)$ and
$A \cap \overline{C}$ are nonempty.

A simple polygonal cycle is said to have \emph{zero area}, if $\area(C) \leq \eps$,
for a sufficiently small given $\eps>0$.
Similarly, a polyhedron $P$ is is said to have \emph{zero volume}, if $\vol(P) \leq \eps$,
for a sufficiently small given $\eps>0$.

For $r>0$, let $B(r)$ denote the ball (\ie, disk in the plane) of radius $r$.
For two convex bodies, $A$ and $B$ let $A+B$ denote their \emph{Minkowski sum},
namely $A+B = \{a+b \ | \ a \in A, b \in B\}$. 

\paragraph{Preliminaries and related work.}
Let $S$ be a finite set of points in the plane.
According to an old result of Few~\cite{Fe55}, the length of a minimum spanning path
(resp., minimum spanning tree) of any $n$ points in the unit square is at most $\sqrt{2n}+7/4$
(resp.,  $\sqrt{n}+7/4$). Both upper bounds are constructive; for example, the construction of a
short spanning path works as follows. Lay out about $\sqrt{n}$ equidistant horizontal lines,
and then visit the points layer by layer,
with the path alternating directions along the horizontal strips.
In particular, the length of the minimum spanning tree of any $n$ points in the unit square
is bounded from above by the same expression. 
An upper bound with a slightly better multiplicative constant for a path was
derived by Karloff~\cite{Ka89}. 
Fejes T\'oth~\cite{Fe40} had observed earlier that for $n$ points of a regular hexagonal
lattice in the unit square, the length of the minimum spanning path is
asymptotically equal to $(4/3)^{1/4}\sqrt{n}$, where $(4/3)^{1/4} = 1.0745\ldots$.
As such, the maximum length of the minimum spanning tree of any $n$ points in the unit square
is $\Theta(\sqrt{n})$, for a small constant (close to $1$).
The $O(\sqrt{n})$ upper bound also holds for points in a convex polygon of diameter $O(1)$,
in particular for $n$ points in a rectangle of diameter $O(1)$.
In every dimension $d\geq 3$, Few showed that the maximum length of a shortest path
(or tree) through $n$ points in the unit cube is $\Theta(n^{1-1/d})$;
the $O(n^{1-1/d})$ upper bound is again constructive and extends to rectangular boxes
of diameter $O(1)$.

The topic of ``separation'' has appeared in multiple interpretations;
here we only give a few examples:~\cite{AFK85,CG16,CDKW05,FMP90,GIK02,HM91,HNRS01}.
Some results on watchman tours relying on Few's bounds can be found in~\cite{DT12b};
others can be be found in~\cite{AMP03}.
For instance, in the problem of finding a separating cycle for a given set of segment pairs,
that we study here, it is clear that the edges of the cycle must hit all of the given segments.
As such, this problem is related to the classic problem of hitting a set of segments
by straight lines~\cite{HM91}. In a broader context, coloring of geometric hypergraphs
has been studied, \eg, in~\cite{Sm07}.

\section{Separating cycles for graphs and hypergraphs} \label {sec:main}

By adapting results on hypergraph $2$-colorability to a geometric setting, we obtain
the following.
\begin{theorem} \label{thm:general}
 Let $H=(V,E)$ be a geometric hypergraph on points in the plane with no singleton edges.
 Then $H$ admits a separating cycle if and only if $H$ is $2$-colorable.
\end{theorem}
\begin{proof}
For the direct implication, assume that $C$ is a separating cycle: then for each $A \in E$,
both $A \cap \ring{C}$ and $A \cap \overline{C}$ are nonempty.
Color the points in the interior of $C$ by red and those in its exterior by blue.
As such, the hypergraph $H$ is $2$-colorable.

We now prove the converse implication. Let $V= R \cup B$ be a partition of the points into red
and blue points, such that no edge in $E$ is monochromatic. We construct a simple polygonal
cycle containing only the red points in its interior. To this end, we first compute a
minimum spanning tree $T$ for the points in $R$; $T$ is non-crossing~\cite[Ch.~6]{PS85},
however there could be blue points contained in edges of $T$. Replace each such edge $s$
with a two-segment polygonal path $\widetilde{s}$ connecting the same pair of points and
lying very close to the original segment, and so that $\widetilde{s}$ is not incident
to any other point.

The resulting tree, $\widetilde{T}$ is still non-crossing and spans all points in $R$.
By doubling the edges of $\widetilde{T}$ and adding short connection edges, if needed,
construct a simple polygonal cycle $C$ of \emph{zero area} that contains it and lies
very close to it; as such, $C$ contains all red points and none of the blue points, as required. 
\end{proof}

Since hypergraph $2$-colorability is NP-complete~\cite{GJ79}, Theorem~\ref{thm:general}
yields the following.

\begin{corollary}
  Given a geometric hypergraph $H=(V,E)$ on points in the plane with no singleton edges,
  the problem of deciding whether $H$ admits a separating cycle is NP-complete. 
\end{corollary}

We next present an approximation algorithm for computing a shortest separating cycle
of a geometric graph. A key fact in our algorithm is the following observation.
\begin{lemma} \label{lem:bipartite}
  Let $G$ be connected bipartite graph. Then (apart from a color flip), 
  $G$ admits a unique $2$-coloring. 
\end{lemma}
\begin{proof}
Recall that a graph is bipartite if and only if it contains no odd cycle~\cite[Ch.~3.3]{Ju99}.
Consider an arbitrary vertex $s$ and color it red. Then the color of any other vertex, say $v$, 
is uniquely determined by the parity of the length of the shortest path from $s$ to $v$
in $G$: red for even length and blue for odd length. Indeed, the vertices are colored alternately
on any path, and since any cycle has even length, all lengths of paths from $s$ to $v$
have the same parity, as required.
\end{proof}

Let $G=(V,E)$ be the input geometric graph no isolated vertices, where $|V|=n$, $|E|=m$.
Let $G_1,\ldots,G_k$ denote the connected components of $G$, where $G_i=(V_i,E_i)$,
for $i=1,\ldots,k$. 

\begin{theorem} \label{thm:graph2}
{\rm (i)}~Given a geometric graph $G=(V,E)$, a separating cycle (if it exists) can be computed
  in $O(m+n \log{n})$ time, where $|V|=n$, $|E|=m$. 
  {\rm (ii)}~Further, a $O(\sqrt{n})$-approximation of the shortest separating cycle can be found
  in polynomial time.
\end{theorem}
\begin{proof}
(i)~The graph is first tested for bipartiteness and the input instance is declared
  infeasible if the test fails (by Theorem~\ref{thm:general}). This test
  takes $O(m+n)$ time; see, \eg, \cite[Ch.~3.3]{Ju99}. We subsequently assume that
  $G$ is bipartite, with vertices colored by red and blue: $V=R \cup B$. Then the algorithm
  constructs a plane spanning tree $T$ of the red points
  (for instance, a minimum spanning tree or a strictly monotone path),
  and outputs a simple cycle by doubling its edges and avoiding the blue points on its edges
  by bending those edges as indicated in the proof of Theorem~\ref{thm:general}.
  
  To this end, the following parameters are computed: For each red point $r \in R$,
  $\delta_1(r)$ is the minimum distance to a blue point. 
  For each edge $e$ of $T$, $\delta_2(e) \geq 0$ is the minimum distance from $e$ to a blue point
  ($\delta_2(e)=0$ if $e$ is incident to a blue point);
  and $\delta_3(e)>0$ is the minimum nonzero distance from $e$ to a blue point
  ($\delta_3(e)=\infty$ if no blue point is close to $e$, as explained below).
  The set of values $\delta_2(e),\delta_3(e)$, $e \in T$, are used for doubling $T$,
  and the set of values $\delta_1(r)$, $r \in R$, are used to determine the separating cycle
  in the vicinity of red vertices; here we omit the details.
  The set of values $\delta_1(r)$, $r \in R$, and $\delta_2(e),\delta_3(e)$, $e \in T$,
  can be determined using point location for the blue points (as query points)
  in a planar triangulated subdivision containing the edges of $T$,
  all in $O(n \log{n})$ time~\cite[Ch.~6]{BCKO08}.
  The overall time complexity of the algorithm is $O(m+n \log{n})$.
  
(ii)~The algorithm above is modified as follows; the first step is the same
  bipartiteness test. The algorithm $2$-colors the vertices
  in each connected component by red and blue: $V_i= R_i \cup B_i$, for $i=1,\ldots,k$.
  By Lemma~\ref{lem:bipartite}, the $2$-coloring of each component is unique
  (apart from a color flip). 
  The initial coloring of a component may be subsequently subject to a color flip
  if the algorithm so later decides. Obviously, the coloring of each component is
  done independently of the others.

  Then, the algorithm guesses the diameter of $\opt$, as determined by one of the ${n \choose 2}$
  pairs of points in $V$ (by trying all such pairs).
  In each iteration, the algorithm may compute a separating cycle and record its length;
  the shortest cycle found in the process will be output by the algorithm; some iterations
  may be abandoned earlier, without the need for this calculation. 
  
  Consider the iteration in which the guess is correct, with pair $a,b \in V$;
  we may assume for concreteness that $ab$ is a horizontal segment of unit length;
  refer to Fig~\ref{fig:f2}. As such, we have that $\len(\opt) \geq 2|ab|=2$. 
  In this iteration, the algorithm computes a separating cycle whose length is bounded from above by
  $O(\sqrt{n})$. First, the algorithm computes a rectangle $Q$ of unit width and height
  $\sqrt{3}$ centered at the midpoint of $ab$. By the diameter assumption, $\opt$ is contained
  in $Q$. In the next step the algorithm computes a separating cycle $C$
  containing only red points in $Q$ in its interior (however, the initial
  coloring of some of the components may be flipped, as needed).
  Note that any separating cycle must contain for each component either all red points or all
  blue points but not a mix of two colors. 
  By Lemma~\ref{lem:bipartite}, the coloring of each component is unique (modulo a color flip)
  and so for each component at least one of its color classes is entirely contained in $Q$.
  As such, all points in $V$ not contained in $Q$ can be discarded from further consideration.

\begin{figure}[htbp]
\centering\includegraphics[scale=0.65]{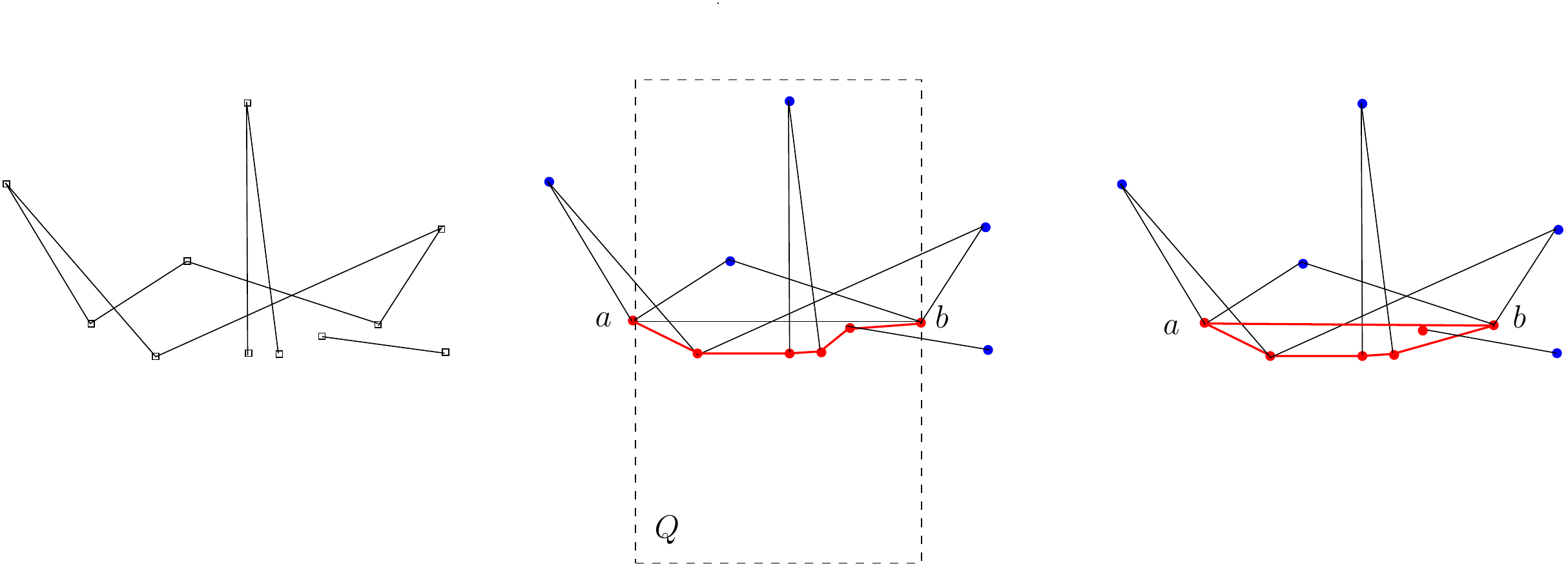}
%\centering\includegraphics[scale=0.65]{f2.eps}
\caption{Left: input bipartite graph. Center: a separating cycle can be computed from the
  MST of the red points (after color flips).
  Right: a shortest separating cycle.}
\label{fig:f2}
\end{figure}

  Each of the components $G_i$, $i=1,\ldots,k$ is checked against this containment condition:
  if a component is found where neither of its two color classes lies in $Q$, the algorithm
  abandons this iteration (and assumed diameter pair, $ab =\diam(\opt)$). 
  For each component $G_i$: (i)~if $R_i \subset Q$, then the coloring of this component remains
  unchanged, regardless of whether $B_i \subset Q$ or $B_i \not\subset Q$.
  (ii)~if $R_i \not\subset Q$ and $B_i \subset Q$, then the coloring of this component is flipped:
  $R_i \leftrightarrow B_i$, so that $R_i \subset Q$ after the color flip.

Once the recoloring of components is complete, the algorithm computes a minimum spanning tree $T$
of the red points in $Q$. Its length is bounded from above by the length of the spanning tree
computed by Few's algorithm. Since the number of red points does not exceed $n$,
we have $\len(T) =O(\sqrt{n})$.
Finally, $T$ is converted into a separating cycle $C$ by a factor of at most $2+\eps$
increase in length, for any given $\eps>0$, as in the proof of part~(i).
Recalling that $\len(\opt) \geq 2$, it follows that $C$ is a $O(\sqrt{n})$-factor approximation
of a shortest separating cycle.
\end{proof}

An instance on which our algorithm---as well as that of Arkin~\etal~\cite{AGH+16}---performs
poorly appears in Fig.~\ref{fig:f13}.
\begin{figure}[hbtp]
\centering\includegraphics[scale=0.7]{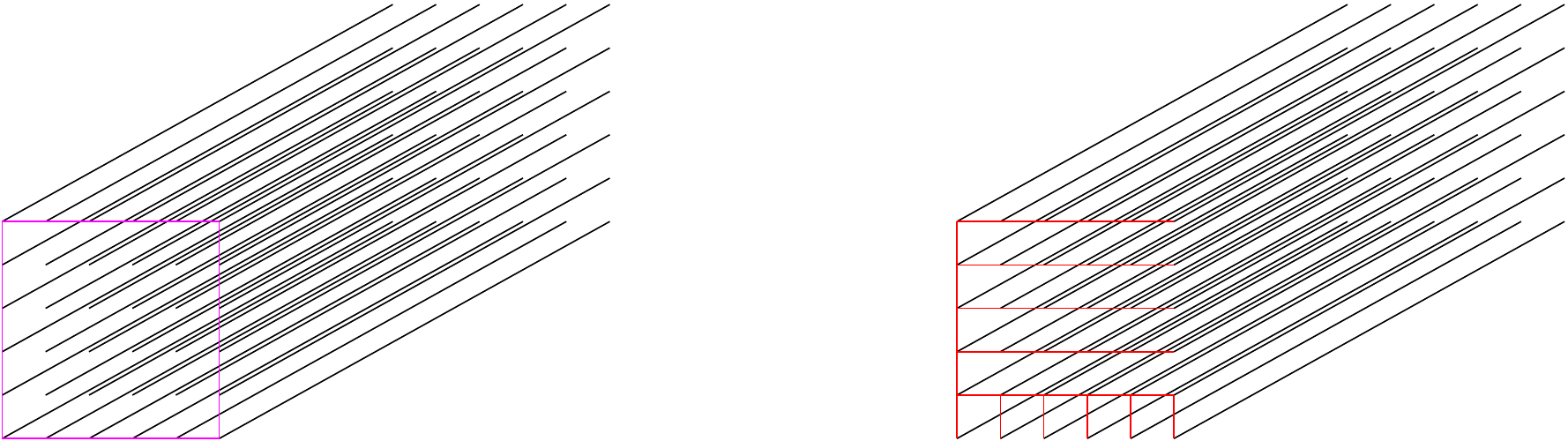}
\caption{An instance for which the $O(\sqrt{n})$-approximation ratio is tight;
  the left and respectively the right segment endpoints are each enclosed in a unit square
  and form the vertices of a $\sqrt{n} \times \sqrt{n}$ uniform grid.
  Left:~a~shortest separating cycle. Right:~a minimum spanning tree of the red points
(as basis for the separating cycle constructed by the algorithm).}
\label{fig:f13}
\end{figure}

\section{Separating cycles for matchings in convex position in the plane} \label {sec:convex}

A matching of $n$ point pairs is said to be in \emph{convex position} if the $2n$ points are
in convex position. In this section we develop a polynomial time approximation scheme~(PTAS)
for this setting; given $n$ point pairs in convex position and $\eps>0$, the algorithm computes a
$(1+\eps)$-approximation of a shortest separating cycle. Denoting an optimal solution by
$\opt$, note that $\opt$ is a convex polygon with $n$ vertices. Moreover, observe that
a shortest separating cycle is a shortest TSP tour for the set of neighborhood pairs
$\{p_i,q_i\}$, $i=1,\ldots,n$; see~\cite[Sec.~7.4]{Mi00} for an overview of the traveling
salesman problem (TSP).
\begin{theorem} \label{thm:convex}
  Given a set of $n$ point pairs in convex position, a $(1+\eps)$-approximation of a shortest
  separating cycle can be computed in time $n^{O(\eps^{-1/2})}$. 
\end{theorem}

We need the following technical lemma for convex polygon approximation.
Our lemma is clearly of independent interest; while it answers a basic question,
we could not find such a result in the literature; there exist however related results,
see, \eg,~\cite{FMRWY92}.
\begin{lemma} \label{lem:convex}
  Given a convex polygon $P$ and any $\eps>0$, there exists
  a subpolygon $Q \subset P$ with $O(1/\sqrt{\eps})$ vertices, such that
  $P \subset \conv(Q) + B(\eps \cdot \diam(P))$. Apart from the constant factor,
  this bound cannot be improved.
\end{lemma}
\begin{proof}
Observe that the number of vertices of $Q$ does not depend on the number of vertices of $P$,
it only depends on $\eps$. We will assume without loss of generality that $\diam(P)=1$. 
Let $P=p_1,p_2, \ldots,p_m$ be the vertices of $P$ labeled clockwise.

First construct  a subpolygon $R \subset P$ iteratively. Set $i \gets 1$ and include $p_i$ into $R$.
Scan $P$ clockwise until we find the first vertex, $p_j$ such that at least one
vertex among $p_{i+1},\ldots, p_{j-1}$ is at distance at least $\eps$ from the chord $p_i p_j$.
Include $p_j$ into $R$. Set $i \gets j$ and continue in the same manner, 
scan $P$ clockwise until we find the first vertex, $p_j$ such that at least one
vertex among $p_{i+1},\ldots, p_{j-1}$ is at distance at least $\eps$ from the chord $p_i p_j$,
and so on. Suppose that the scanning ends after $r$ phases; then $R$ has either $r$ or $r+1$ vertices;
see Fig.~\ref{fig:f16}\,(left). Note that the last side of $R$ can be unconstrained, \ie, with 
no guarantee of a vertex of $P$ at distance at least $\eps$ from it. 
\begin{figure}[hbtp]
\centering\includegraphics[scale=0.65]{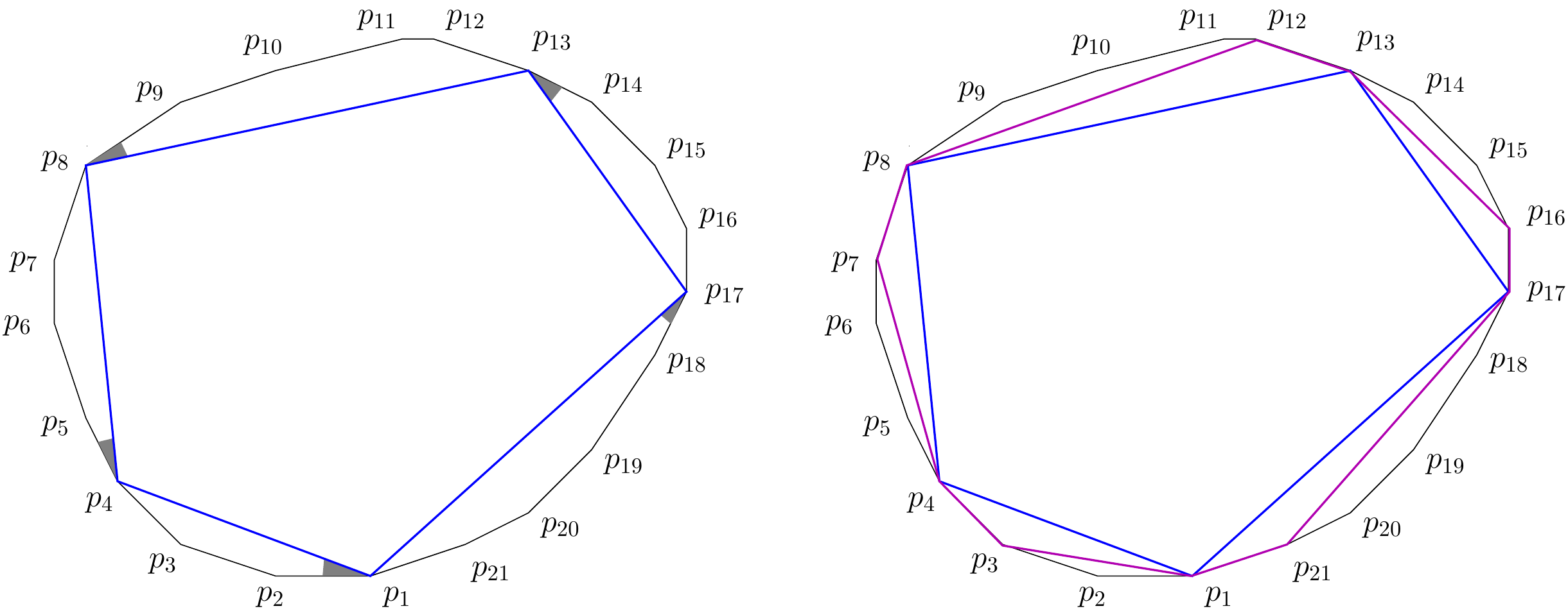}
%\centering\includegraphics[scale=0.65]{f16.eps}
\caption{Convex polygon approximation; here $m=21$, $r=5$.
  $R=\{p_1,p_4,p_8,p_{13},p_{17}\}$ is drawn in blue lines. 
  $Q=\{p_1,p_3,p_4,p_7,p_8,p_{12},p_{13},p_{16},p_{17},p_{21}\}$ is drawn in magenta lines.
  The angles $\alpha(\sigma)$ are shaded (left side of figure):
  $\alpha(p_1 p_4)=\angle{p_2 p_1 p_4}$, $\alpha(p_4 p_6)=\angle{p_5 p_4 p_8}$, and so on.}
\label{fig:f16}
\end{figure}

Each side of $R$ is a chord or side of $P$. A side of $R$ is said to be \emph{short}
if its length is at most $\sqrt{\eps}$ and \emph{long} otherwise. 
Write $r=r_1 + r_2$, where $r_1$ and $r_2$ are the number of short and long sides of $R$,
respectively. 

Since $\diam(P)=1$, we have $\per(P) \leq \pi$ by the classic isoperimetric inequality.
Since $R$ is a subpolygon of $P$, we have $\per(R) \leq \per(P)$ by the triangle inequality;
thus $\per(R) \leq \pi$. This further implies that
$r_2 \leq \per(R)/\sqrt{\eps} \leq \pi/\sqrt{\eps}$.

When scanning $P$ clockwise, let $\sigma$ be a side of $R$ that is a nontrivial chord of $P$;
let $\alpha(\sigma)$ denote the angle made by the chord with the first clockwise edge of $P$
on the boundary of $P$. By convexity, the angles $\alpha(\sigma)$ corresponding to all nontrivial chords of $P$
that are edges in $R$ are pairwise non-overlapping if their apices are placed at a common point;
see Fig.~\ref{fig:f16}\,(left). As such,
\begin{equation} \label{eq:sum}
  \sum_{\sigma \in R} \alpha(\sigma) \leq 2\pi.
\end{equation}
Recall that $\tan{x} \leq 2x$, for $0 \leq x \leq \pi/3$. 
Consider any short side $\sigma \in R$ that is a chord of $P$; we thus have
\begin{equation} \label{eq:alpha}
2 \alpha(\sigma) \geq \tan{\alpha(\sigma)} \geq \frac{\eps}{|\sigma|}
\geq \frac{\eps}{\sqrt{\eps}} =\sqrt{\eps}, \text{~~~~or~} \alpha(\sigma) \geq \pi/3.
\end{equation}
It follows from~\eqref{eq:sum} and~\eqref{eq:alpha} that
$$ r_1 \leq \frac{2\pi}{\sqrt{\eps}/2} + \frac{2\pi}{\pi/3} = \frac{4\pi}{\sqrt{\eps}} +6. $$
Consequently,
$$ r=r_1 + r_2 \leq \frac{4\pi}{\sqrt{\eps}} +6 + \frac{\pi}{\sqrt{\eps}}
= \frac{5\pi}{\sqrt{\eps}} + 6. $$

To obtain $Q$ we further subdivide each polygonal arc of $P$ spanned by edges of $R$
(with the possible exception of the last arc)  as follows: the arc $p_i p_j$ is subdivided at $p_{j-1}$,
that is, this vertex is included along with $p_i$ and $p_j$ into $Q$. Observe that $Q$ has at most
$$ 2(r+1) \leq \left \lceil \frac{10 \pi}{\sqrt{\eps}} +14 \right \rceil = O(\eps^{-1/2}) $$
vertices. 
(In particular, $Q$ has at most $\frac{32}{\sqrt{\eps}}$ vertices, if $\eps>0$ is sufficiently small.)
In addition, by construction, $Q$'s enlargement contains $P$: recall that vertices in $R$ are
iteratively chosen so that the corresponding arcs of $P$ are minimally breaking the inclusion
property stated in the lemma, and so the subdivision of each arc described above will satisfy
this property. In particular, each vertex of $P$ is at distance at most $\eps$ from the corresponding
side of $Q$. 

\medskip
The case of a regular polygon shows that the bound is tight. Let $P$ be a regular
$m$-gon inscribed in a circle of unit radius, and $Q \subset P$ be a subpolygon, such that
every vertex of $P$ is at distance at most $\eps$ from the corresponding side of $Q$.
Let $2 \beta$ be the center angle spanned by the longest edge of $Q$; for simplicity, 
assume that the number of vertices of $P$ on the corresponding polygonal arc is odd
(the other case is similar). 
The distance condition requires $1-\cos \beta \approx \beta^2/2= \eps$, which 
solves to $\beta =\sqrt{2 \eps}$. This implies that the number of vertices of $Q$ is
$\Omega(\beta^{-1}) =\Omega(\eps^{-1/2})$, as required.
\end{proof}

\paragraph{Remark.} If $P \subset \conv(Q) + B(\eps \cdot \diam(P))$,  \ie,
$Q$ is a suitable approximation as required in Lemma~\ref{lem:convex},
then $\diam(Q) \geq \diam(P)/(1+2\eps)$, and consequently,
$$ P \subset \conv(Q) + B(\eps (1+2\eps) \cdot \diam(Q)). $$
Indeed, $ \eps \cdot \diam(P) \leq \eps (1+2\eps) \cdot \diam(Q)$,
from which the inclusion follows. 

\medskip
Assume now that $P$ is the \emph{unknown} convex polygon $\opt$.
Set $k =\lceil \frac{10 \pi}{\sqrt{\eps}} +14 \rceil$.
By convexity, $P$ contains exactly one point from each pair; as such, $|P|=n$. 
The algorithm finds a subpolygon $Q \subset P$ satisfying the property in Lemma~\ref{lem:convex}
by generating all subpolygons $Q$ with at most $k$ vertices.
It does so by generating all ${n \choose \leq k}$ subsets of at most $k$ segments;
for a subset of $i \leq k$ segments, it goes through all $2^i$ possible endpoint
selections (one point from each pair).
The reason is that the shortest separating cycle of some $i \leq k$ endpoints
(one point from each pair) may be an infeasible candidate for $Q$; see~Fig.~\ref{fig:f17},
but generating shortest separating cycles for each choice of $i \leq k$ endpoints
will yield a feasible candidate, as required by Lemma~\ref{lem:convex}. 
\begin{figure}[hbtp]
\centering\includegraphics[scale=0.99]{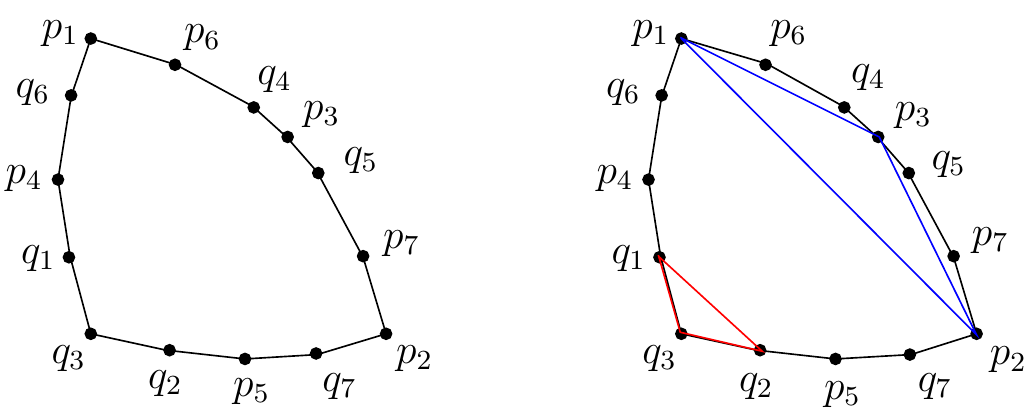}
%\centering\includegraphics[scale=0.99]{f17.eps}
\caption{Left: a set of seven pairs with points in convex position.
  Right: the shortest separating cycle for the pairs $\{p_i,q_i\}$, $i=1,2,3$, 
  $Q=\{q_1,q_2,q_3\}$ (drawn in red) is an infeasible candidate,
  \ie, $\conv(Q) + B(\eps \cdot \diam(\opt))$
  does not contain the optimal convex polygon $P=\opt =\{p_1,p_6,q_4,p_3,q_5,p_7,p_2\}$.
  On the other hand, $Q=\{p_1,p_3,p_2\}$ (drawn in blue) is a valid candidate.}
\label{fig:f17}
\end{figure}

\bigskip
\noindent{\bf Algorithm~A1.} 
\begin{itemize} \itemsep 1pt
\item[] {\sc Step 1:} Set $k = \left \lceil \frac{10 \pi}{\sqrt{\eps}} +14 \right \rceil$.
  Generate all $\sum_{i=1}^k {n \choose i}$ subsets of at most $k$ segments.

\item[] {\sc Step 2:} For a subset of size $i$, go through all $2^i$ endpoint selections 
(one point from each pair).

\item[] {\sc Step 3:} For a subset $Q$ of $i \leq k$ points as described above,
  if $\conv(Q) + B(\eps (1+2\eps) \cdot \diam(Q))$ contains at least
  one point from each of the $n$ input pairs, keep one such point from each pair
  (if both points of a pair are enclosed, choose one arbitrarily);
  then compute the perimeter of the convex polygon made by these $n$ points,
  \ie, the length of the separating cycle.  Otherwise skip this subset.
  
\item[] {\sc Step 4:} Return the cycle of minimum length from among those
  computed in {\sc Step 3}. 

\end{itemize}

The running time of Algorithm~A1 is determined by the number of candidates examined,
namely 
$$ \sum_{i=1}^k {n \choose i} 2^i \leq n^k =n^{O(\eps^{-1/2})}. $$
The algorithm correctness follows from Lemma~\ref{lem:convex}; indeed,
$P=\opt$ is contained in
$$ \conv(Q) + B(\eps \cdot \diam(P)) \subset \conv(Q) +B(\eps (1+2\eps) \cdot \diam(Q)), $$
namely the enlargement of one of the candidates $Q$ that are generated.
The length of the separating cycle that is returned is bounded from above by the perimeter
of the enlargement. We employ the following standard fact.

\begin{lemma} \label{lem:cauchy}
Let $Q$ be a planar convex body.
Then $\per(\conv(Q) + B(r)) =\per(Q) + 2\pi r$. 
\end{lemma}
\begin{proof}
Let $w(\alpha)$ denote the width of $Q$ in direction $\alpha$, \ie, the
minimum width of a strip of parallel lines enclosing $Q$, whose lines
are orthogonal to direction $\alpha$.
According to Cauchy's surface area formula~\cite[pp.~283--284]{PA95},
for any planar convex body $Q$, we have
\begin{equation} \label{eq:cauchy}
\int_0^{\pi} w(\alpha) \intd \alpha = \per(Q).
\end{equation}

Observe that the width of $\conv(Q) + B(r)$ in direction $\alpha$ is $w(\alpha) +2r$, for any
$\alpha \in [0,\pi)$. Using the stated formula we have
  $$  \per(\conv(Q) + B(r)) = \int_0^{\pi} (w(\alpha) +2r) \intd \alpha
  = \int_0^{\pi} w(\alpha) \intd \alpha + 2 \pi r 
  = \per(Q) + 2 \pi r, $$
as required.
\end{proof}

By Lemma~\ref{lem:cauchy}, the length of the separating cycle $C$ returned by Algorithm~A1
is bounded from above (for $\eps$ sufficiently small) by 
\begin{align*}
  \len(C) &\leq \per(\conv(Q) + B(\eps (1+2\eps) \cdot \diam(Q))\\
  &= \per(Q) + 2 \pi \eps (1+\eps) \cdot \diam(Q)\\
  &\leq \per(P) + 2 \pi \eps (1+\eps) \cdot \diam(P) \\
  &\leq \per(P) + 2 \pi \eps (1+\eps) \cdot \per(P) /2 \\
  &= (1+ \pi \eps (1+\eps)) \,  \per(P) = (1 + \pi \eps (1+\eps)) \, \len(\opt) \\
  &\leq (1 + 4\eps) \, \len(\opt). 
\end{align*}
The required approximation follows by rescaling $\eps$,
and this completes the proof of Theorem~\ref{thm:convex}.

\section{Concluding remarks} \label{sec:conclusion}

\paragraph{Remark 1.}
If the input is a set of pairs so that the corresponding graph is bipartite, it admits a separating cycle
by Theorem~\ref{thm:general}. (If the corresponding graph is not bipartite, no separating cycle exists.)
Similarly, if the input is a $2$-colorable hypergraph, it admits a separating cycle.
For illustration, we recall some common instances of $2$-colorable hypergraphs.
A hypergraph $H=(V,E)$ is called $k$-uniform if all $A \in E$ have $|A|=k$. 
A~random $2$-coloring argument gives that any $k$-uniform hypergraph with
fewer than $2^{k-1}$ edges is $2$-colorable~\cite[Ch.~1.3]{AS15};
as such, by Theorem~\ref{thm:general}, it admits a separating cycle.
Slightly better bounds have been recently obtained; see~\cite[Ch.~3.5]{AS15}.
Similarly, let  $H=(V,E)$ be a hypergraph in which every edge has size at least $k$
and assume that every edge $A \in E$ intersects at most $\Delta$ other edges,
\ie, the maximum degree in $H$ is at most $\Delta$. If $e(\Delta+1) \leq 2^{k-1}$
(here $e=\sum_{i=0}^\infty 1/i!$ is the base of the natural logarithm),
then by the Lov\'asz Local Lemma, $H$ can be $2$-colored~\cite[Ch.~5.2]{AS15}
and so by Theorem~\ref{thm:general}, it admits a separating cycle;
moreover, if a $2$-coloring is given, it can be used to obtain a separating cycle.
While testing for $2$-colorability can be computationally expensive in a general setting
(recall that hypergraph $2$-colorability is NP-complete~\cite{GJ79}), 
it can be always achieved in exponential time.

\paragraph{Remark 2.}
Theorem~\ref{thm:graph2} generalizes to $3$-dimensional polyhedra.
A polyhedron in $3$-space is a simply connected solid bounded by piecewise linear
$2$-dimensional manifolds. The \emph{perimeter} $\per(P)$ of a polyhedron $P$ is the total length
of the edges of $P$ (as in~\cite{DT12b}).

For part~(i), a method similar to that used in the planar case can be used to construct
a separating polyhedron in $\RR^3$ (or $\RR^d$). However, since computing minimum spanning trees
in $\RR^3$ is more expensive~\cite[Ch.~9]{Ep00}, we employ a slightly different approach
(in particular, this approach is also applicable to the planar case).
We may assume a coordinate system so that no pair of points have the same $x$-coordinate.
First, the points in $V$ are colored by red or blue as a result of the
bipartiteness test, in $O(m+n)$ time.
The algorithm then computes a (spanning tree of the red points in the form of a) 
$x$-monotone polygonal path $P$ spanning the red points; this step takes $O(n\log{n})$ time.
From $P$, it then obtains a $x$-monotone polygonal path $\widetilde{P}$
spanning the red points and not incident to any blue point
($P=\widetilde{P}$ if no blue points are incident to edges of $P$);
$\widetilde{P}$ is constructed in $O(n\log{n})$ time.

To this end, $P$ and all blue points are projected onto the $xoy$ plane.
Let $\sigma(\cdot)$ denote the projection function. Note that $\sigma(P)$ is $x$-monotone 
and that the projection $\sigma(b)$ of a blue point $b$ can be incident to at most one edge
of $\sigma(P)$; given $b$, such an edge can be determined in $O(\log{n})$ time by binary search.
Checking the projection points $\sigma(b)$ against corresponding edges of $\sigma(P)$
allows for testing whether the original edges of $P$ are incident to the respective
blue points. Further, this test allows replacing each such edge $s$
with a two-segment polygonal path $\widetilde{s}$ connecting the same
pair of points and lying very close to the original segment, and so that $\widetilde{s}$ is not
incident to any other point; see Fig~\ref{fig:f12}.
Such a replacement can be executed in $O(1)$ time per edge.
Finally the algorithm computes a polyhedron of \emph{zero volume} that contains $\widetilde{P}$;
as such, the polyhedron contains all red points but no blue points;
this step takes $O(n\log{n})$ time. %Some details are omitted. 
\begin{figure}[htbp]
\centering\includegraphics[scale=0.85]{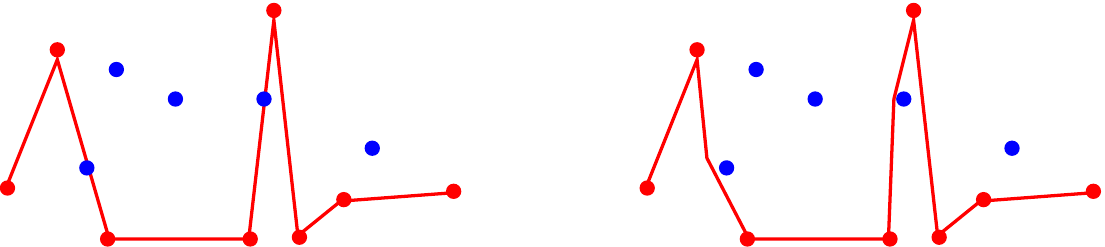}
\caption{Replacing two edges of the $x$-monotone path spanning the red points.}
\label{fig:f12}
\end{figure}

For part~(ii),  instead of a rectangle based on segment $ab$ as an assumed diameter pair,
the algorithm works with a rectangular box where $ab$ is parallel to a side of the box
and is incident to its center. The upper bound on the perimeter of the separating polyhedron
follows from Few's bound mentioned in the preliminaries: it is roughly three times the length
of a shortest path (or tree) spanning the red points. 
\begin{theorem} \label{thm:graph3}
{\rm (i)}~Given a geometric graph $G=(V,E)$ in $\RR^3$, a separating polyhedron (if it exists)
  can be found in $O(m+ n\log{n})$ time, where $|V|=n$, $|E|=m$. 
{\rm (ii)}~Further, a $O(n^{2/3})$-approximation of a separating polyhedron of minimum perimeter
  can be found in polynomial time.
\end{theorem}

We offered a characterization of geometric hypergraphs that admit separating cycles
and gave several approximation algorithms. We conclude with the following questions
regarding the shortest separating cycle in the plane.

\begin{enumerate} \itemsep 1pt
\item Can the $O(\sqrt{n})$ approximation factor for the general version of the problem
  be improved?
\item Can sharper results be obtained for plane (noncrossing) geometric graphs?
  For the case of a plane matching?
\item What is the computational complexity of the problem
  for matchings in convex position? Does the problem admit a polynomial-time
  algorithm?
\end{enumerate}

\paragraph{Acknowledgment.}
The author is grateful to an anonymous reviewer for his careful reading of the manuscript
and pertinent remarks.

\end{document}